%% file: main.tex
\documentclass{article}

\usepackage{hyperref}

\usepackage[utf8]{inputenc} % allow utf-8 input
\usepackage[T1]{fontenc}    % use 8-bit T1 fonts
\usepackage{hyperref}       % hyperlinks
\usepackage{url}            % simple URL typesetting
\usepackage{amsfonts}       % blackboard math symbols
\usepackage{nicefrac}       % compact symbols for 1/2, etc.
\usepackage{microtype}      % microtypography
\usepackage{xspace}
\usepackage{wrapfig}

\pdfoutput=1
\usepackage{graphicx} % more modern
\usepackage[square,sort,comma,numbers]{natbib}
\usepackage{mathtools}
\usepackage{booktabs} % For formal tables
\usepackage{times}
\usepackage{graphicx} % more modern
\usepackage{amsthm}
\usepackage{amssymb}
\usepackage{amsmath}
\usepackage[bottom]{footmisc}

\usepackage{caption}
\usepackage{subcaption}
\usepackage{amsmath}
\usepackage{enumerate}

\usepackage[utf8]{inputenc} % allow utf-8 input
\usepackage[T1]{fontenc}    % use 8-bit T1 fonts
\usepackage{hyperref}       % hyperlinks
\usepackage{url}            % simple URL typesetting
\usepackage{booktabs}       % professional-quality tables
\usepackage{amsfonts}       % blackboard math symbols
\usepackage{nicefrac}       % compact symbols for 1/2, etc.
\usepackage{microtype}      % microtypography

\newcommand{\comment}[1]{}

\newtheorem{theorem}{Theorem}
\newtheorem{lemma}{Lemma}
\newtheorem{proposition}{Proposition}
\newtheorem{remark}{Remark}

\newcommand{\esm}[1]{\ensuremath{#1}}

\newcommand\reals{\mathbb{R}}

\newcommand\V{F}
\newcommand\player{feature\xspace}
\newcommand\players{features\xspace}

\newcommand\game{function\xspace}
\newcommand\games{functions\xspace}

\newcommand\efficiency{efficiency\xspace}

\newcommand\share{attribution\xspace}
\newcommand\shares{attributions\xspace}
 
\newcommand\prefixvalues{Shapley-Taylor indices\xspace}
\newcommand\stv{\prefixvalues}
\newcommand\siv{Shapley interaction indices\xspace}

\newcommand\interaction{\esm{\mathcal{I}^k}}

 \author{
   Kedar Dhamdhere \\
   \and
   Mukund Sundararajan \\ 
   \and
   Ashish Agarwal \\
   Google LLC \\
   Mountain View, USA}

\begin{document}

%\twocolumn[
%\icmltitle{The Shapley Taylor Interaction Index}
\title{The Shapley Taylor Interaction Index}




%\begin{icmlauthorlist}
%\icmlauthor{Kedar Dhamdhere}{goo}
%\icmlauthor{Ashish Agarwal}{goo}
%\icmlauthor{Mukund Sundararajan}{goo}
%\end{icmlauthorlist}

%\icmlaffiliation{goo}{Google LLC}

%\icmlcorrespondingauthor{Mukund Sundararajan}{mukunds@google.com}

%\icmlkeywords{Machine Learning, ICML}

\vskip 0.3in
%]



\maketitle

\begin{abstract}

The \emph{attribution} problem, that is the problem of attributing a model's prediction to its base features, is well-studied. We extend the notion of attribution to also apply to feature interactions.

The Shapley value is a commonly used method to attribute a model's prediction to its base features. We propose a generalization of the Shapley value called Shapley-Taylor index that attributes the model's prediction to interactions of subsets of features up to some size $k$. The method is analogous to how the truncated Taylor Series decomposes the function value at a certain point using its  derivatives at a different point. In fact, we show that the Shapley Taylor index is equal to the Taylor Series of the \emph{multilinear extension} of the set-theoretic behavior of the model.   

We axiomatize this method using the standard Shapley axioms---\emph{linearity}, \emph{dummy}, \emph{symmetry} and \emph{efficiency}---and an additional axiom that we call the \emph{interaction distribution} axiom. This new axiom explicitly characterizes how interactions are distributed for a class of functions that model pure interaction. 

We contrast the Shapley-Taylor index against the previously proposed Shapley Interaction index (cf. ~\cite{Grabisch99}) from the cooperative game theory literature. We also apply the Shapley Taylor index to three models and identify interesting qualitative insights.
\end{abstract}

\input{intro.tex}

\input{method.tex}

\input{axioms.tex}

\input{evaluation.tex}

% Bibliography
\bibliographystyle{acm}
\bibliography{main}

\input{appendix}

\end{document}

%% file: intro.tex
\section{Introduction}
\subsection{Motivation}
\label{sec:summary}

 There is considerable literature on feature importance/attribution for deep networks(cf.~\cite{BSHKHM10,SVZ13,SGSK16,BMBMS16,SDBR14,LL17, STY17,Lundberg2017AUA}). The basic idea of attribution is to distribute the prediction score of a model for a specific input to its base features; the attribution to a base feature can be interpreted as its contribution to the prediction. For instance, when attribution is applied to a network that predicts the sentiment associated with a paragraph of text, it quantifies the influence of every word in the text on the network's score. This can be used, for instance, to tell if the model bases its prediction on words that connote a protected category like a specific race/gender/religion. This would be indicative of the model possibly being biased along the protected dimension. Attribution/feature importance for Deep Networks has been applied to a variety of real world applications, for instance in health, drug-discovery, machine translation, natural language tasks, recommendation systems etc. Thus, attributions are quite useful despite their simple form; notice they don't reveal the logic of the network beyond base feature importance.

   In this paper, we take a step towards making attributions a somewhat richer form of explanation by identifying the importance of \emph{feature interactions}, either pairwise or of higher orders. We would like to identify to what extent a set of features exert influence in conjunction as opposed to independently. We expect the study of interactions to be fruitful. Deep networks are likely to have an abundance of strong feature interactions, because they excel at creating higher-order representations (e.g. filters) out of the base features (e.g. pixels). We also expect the study of interactions to be critical for the tasks that cannot be performed by features acting independently. We study two such tasks in Section~\ref{sec:applications}. In the sentiment analysis task, \emph{negation}, should manifest as an interaction between the negation word (e.g. not) and the sentiment bearing word (e.g. good or bad) that it modifies. If such an interaction is not detectable, then the network needs fixing. Another example is reading comprehension, i.e., question answering about paragraphs of text. A good model will match question words to certain words/phrases in the paragraph, and these matches should manifest as interactions between those words. 
   
   Before we describe our contributions, we mention some related work on feature interactions besides the attribution literature briefly described above.

\subsection{More Related Work}

\subsubsection{Shapley Value, Shapley Interaction Value}
Some of the deep network attribution literature (cf.~\cite{STY17, Lundberg2017AUA, Lundberg18, Datta2015Influence, Datta2016Algorithmic, Strumbelj2010, Sliwinski2019}) is built on prior work in cooperative game theory, specifically, the the Shapley value (\cite{Shapley53}) and its continuous variant~\cite{AS74}. These prescribe a way to distribute the value of a game among its players.\footnote{Games are analogous to models, players to features, and the shares to the feature importance.} Shapley values have been used to study global feature/variable importance in statistics (cf.~\cite{Owen2014, Cohen2005}). 
The work most closely related to ours is the Shapley interaction value, first proposed by \cite{Owen72} to study pairwise interactions between players. \cite{Grabisch99} generalized it to study interactions of higher orders, and provided an axiomatic foundation. \cite{Lundberg18} applied it to studying feature interactions in trees. We provide comparison against the Shapley interaction value.
 
\subsubsection{Interactions in Machine Learning} 
 
 It is hard to describe all of the vast literature on interactions in machine learning and statistics. Most of this literature is focused on \emph{global} feature importance, i.e., important interactions across the data set. In contrast, we study feature importance for individual inputs. 
 
 There is, for instance, the classic literature on ANOVA (cf.~\cite{ANOVA}), and the more recent literature on Lasso,~\cite{Lasso}, both of which can be used to quantify the importance of putative interactions. 
 
 We mention some recent deep learning literature: \cite{Tsang} constructs a generalized additive model that mimics the behavior of a deep network by investigating the structure of the inter-layer weight matrices. \cite{Tsang2} forces the weight matrices to be block-diagonal, restricting the type of interactions, by designing the appropriate regularization. \cite{Cui} studies \emph{pairwise} interactions by building deep networks of a specific form and then interpreting the network via its gradients, and averaging appropriately over the data set. \cite{murdoch} combines 
 agglomerative clustering with influence propagation across the layers of the deep-network, to produce a hierarchical decomposition of the input with influence scores for each group; this work, unlike the others in this section, is about feature importance of individual inputs.

\subsection{Model} 

In this section, we formally model the attribution problem for interactions. We have a set of features $N$. The deep network is modeled as a function $\V: 2^N \to \reals$, i.e., we treat the features as discrete variables. Modeling the network as as a function of boolean features simplifies our theory of attribution, i.e., we can investigate the influence of variables via the change in score resulting from removing/ablating the variable. For many problems, treating features as boolean is natural; for a text model, words are either present or absent. In our applications, we will model the absence of a feature by replacing it with a sentinel value (zero embedding vector/out of vocabulary word/average value of the feature across the training data). 

To simplify the notation, we denote the cardinalities of a set $S, T$ etc. using lowercase letters: $s, t$. We omit braces for small sets and write $T \cup i$ instead of $T \cup \{ i \}$ and $T \cup  ij $ instead of $T \cup \{ i,j \}$.

Define
\begin{equation}
\delta_i\V(T) = \V(T \cup i) - \V(T)
\end{equation}
and
\begin{equation}
\label{hessian}
\delta_{ij}\V(T) = \V(T \cup ij) - \V(T \cup i) - \V(T \cup j) + \V(T)
\end{equation}

Here $i, j \not\in T$.\footnote{If features $i$ and $j$ don't interact, this quantity will be zero, because adding $j$ to the set $T$ or to the set $T\cup i$ will have the same influence.} These are discrete equivalents of first and second order derivatives. In general, for $T \subseteq N \setminus S$, we define the discrete derivative with respect to set $S$ as:
\begin{equation}
\label{discrete-derivative}
\delta_S \V(T) = \sum_{W \subseteq S} (-1)^{w-s} \V(W \cup T)    
\end{equation}

We use a fixed number $k$ as the \textbf{order of explanation} to mean that we'll compute the \prefixvalues for subsets of size up to $k$. For instance, $k = 2$ corresponds to computing main effects as well as pairwise interaction effects. For a set $S$  such that $|S| \leq k$, $\interaction_S(v)$ denotes the interaction effect for the set $S$.

\subsection{Our Results}
\label{sec:results}

 Our goal is to axiomatically identify the \emph{best $k$th order explanation}. For instance, when $k=2$, we would like to identify the main effects of each feature and the interaction effects for every pair of features. 
 
 \begin{itemize}
     \item  We propose the \emph{Shapley-Taylor index} to solve this problem (Section~\ref{sec:method}).
     \item  We axiomatize the \stv in Section~\ref{sec:axiom}. We introduce a new axiom called the \emph{interaction distribution} axiom. This axiom explicitly defines how interactions are distributed for a class of functions called \emph{interaction functions}.
     \item We compare the \stv against a previously proposed method called \siv(\cite{Grabisch99}).
     This method is the unique method that satisfies variants of the standard Shapley axioms, namely dummy,  linearity and symmetry (see Section~\ref{sec:axiom} for axiom definitions) and an additional axiom called \emph{the recursive axiom}\footnote{The recursive axiom requires that the interaction index for a pair of \players $i,j$ is equal to the difference in the Shapley value of \player $i$ in a game with \player $j$ omnipresent and the Shapley value of \player $i$ in a game with \player $j$ absent.}. Critically, the \siv do not satisfy efficiency, i.e., the attributions do not sum to the score of the function at the input minus its score at the empty set. Axiomatically, the main contribution of this paper is to replace the recursive axiom with the efficiency axiom; the consequence is that the method of choice goes from being \siv to \stv. Section~\ref{sec:example} contrasts the results of the two methods. We find that the lack of efficiency causes \siv to amplify interaction effects (Section~\ref{sec:linear}), or causes the interaction effects to have seemingly incorrect signs (Section~\ref{sec:majority}).
     \item  In Section~\ref{sec:taylor}, we connect the Shapley-Taylor interaction index to the Taylor series (with a Lagrangian remainder); we show that the Taylor series applied to the so called multilinear extension of the function is equivalent to the Shapley-Taylor index applied to the function.
     \item Though our key contributions and evaluations are mainly theoretical, we demonstrate the applicability of our work in Section~\ref{sec:applications}, which studies models for three tasks (sentiment analysis, random forest regression, and question answering). We identify certain interesting interactions.
 \end{itemize}

%% file: method.tex
\section{\stv}\label{sec:method}

\subsection{Shapley Value}
\label{sec:shapley}

Before we discuss feature interactions, let us revisit the Shapley value. The central concept in the Shapley value is that of the \emph{marginal}, i.e., the change in the function value $\V$ by the addition of a \player, i.e. $\delta_{i}\V(S) = \V(S \cup i) - \V(S)$. In general, for a nonlinear function $\V$, the value of this expression depends on the set $S \subseteq N \setminus \{i\}$ at which we compute the marginal and there are several choices for this set $S$. The Shapley value defines a random process that implicitly prescribes a certain weighting over these sets. Given an ordering of the \players, add the \players one by one in this order. Each \player $i$ gets ascribed the marginal value of adding it to the \players that precede it. The Shapley value of the \player $i$ is the expected value of this marginal over an ordering of \players chosen uniformly at random.

The Shapley Value is known to be the unique method that satisfies four axioms (see Section~\ref{sec:axiom} for formal definitions of the axioms): \emph{efficiency} (the attributions of all the \players sum to the difference between the prediction scores at the input minus that at the empty set.), \emph{symmetry} (symmetric \players receive equal \shares), \emph{linearity} (the sum of the \shares of a \player for two functions is identical to the \share of the \player in the game formed by the sum of the two functions) and the \emph{dummy} (if
the marginal value of a \player is zero for every set $S$, it has a zero \share) axioms. 

\subsection{Definition of \stv}

In this section we define \prefixvalues. Just like the Shapley value, \stv can also be computed as an expectation over orderings chosen uniformly at random. The output of the \stv is more extensive compared to the Shapley value; it includes attributions to interaction terms.  

For an order of explanation $k = 2$, i.e., the index specifies values for individual \players and pairs of \players. The indices for an individual \player $i$ is just the marginal $\delta_{i}\V(\emptyset)$. To compute the indices for pairs, we pick an ordering of the \players, and ascribe to each pair of \players the expression~\ref{hessian} computed at the set $S$ of \players that precede both $i$ and $j$ in the ordering. The Shapley Taylor interaction index for a pair is the expectation of this quantity over an ordering chosen uniformly at random. 

Similarly for the general case, \prefixvalues are defined by random process over orderings of \players. Let $\pi = (i_1, i_2, \ldots, i_n)$ be an ordering. Let $\pi^{i_k} = \{i_1, \ldots, i_{k-1} \}$ be the set of predecessors of
$i_k$ in $\pi$. For a fixed ordering $\pi$ and a set, we define \prefixvalues $\interaction_{S, \pi}(\V)$ as follows:

\begin{equation}\label{eq:permutation}
\interaction_{S, \pi}(\V) = 
\begin{cases}
\delta_{S} \V(\emptyset) & \text{if $|S| < k$} \\
\delta_{S} \V(\pi^S) & \text{if $|S| = k$,} \\
\end{cases}
\end{equation}

Here, $\pi^S$ is defined as $\cap_{i \in S} \pi^i$, the set of elements that precede all of the \players in $S$. Let us briefly discuss the three cases. When the size of the interaction term is strictly less than the order of explanation, its interaction value (for the fixed permutation) is simply equal to the the discrete derivative (Equation~\ref{discrete-derivative}) at the empty set. Notice that this quantity does not depend on the permutation itself.
When the order of approximation is $k=2$, this is just the marginal value of adding the feature to the empty set. When the size of the interaction term is \emph{equal} to the order of explanation, its interaction value (for the fixed permutation) is equal to the discrete derivative at the largest set of elements that precede all of the elements of $S$. when the order of approximation is $k=2$, the discrete derivatives match Equation~\ref{hessian}. 

The \prefixvalues are defined as the expectation of $\interaction_{S, \pi}(\V)$ over an ordering of the $N$ \players chosen uniformly at random:
\begin{equation}\label{eq:randorder}
  \interaction_{S}(\V) = \mathbb{E}_{\pi}(\interaction_{S, \pi}(\V))
\end{equation}

A noteworthy special case is $k=1$. Here the first case of Equation~\ref{eq:permutation} does not occur and the discrete derivatives correspond to marginals. The resulting definition is precisely the well-known Shapley value. \stv thus generalize the Shapley value. 

We have defined the \stv in terms of permutations. The theorem below derives a closed form expression for \stv. This provides a method to compute the \stv. We have given the proof in the Appendix (Section~\ref{sec:stv_expression}).

\begin{theorem}\label{thm:stv_expression}
Let $k$ be the order of explanation. For a set $S \subseteq N$, such that $|S| = k$,
\stv satisfy the following expression:
\begin{equation}
\label{eq:stv_formula}
 \interaction_{S}(\V) = \frac{k}{n} \sum_{T \subseteq N \setminus S} \delta_S \V(T)  \frac{1}{{n-1 \choose t}}   
\end{equation}
\end{theorem}

\begin{remark}
The formula in Theorem~\ref{thm:stv_expression} gives us a way to compute the \stv. It involves computation over all subsets of the feature
set and hence it takes exponential time. In practice, we can trade off accuracy for speed. One way to obtain a fast approximation is to apply Equation~\ref{eq:permutation} over a sample of permutations. This is similar to the Shapley value sampling methods\cite{Maleki2013}.
In Section~\ref{sec:applications}, we use a different approximation. We first identify a subset of features with high attribution (using Shapley values or Integrated Gradients method). We use Theorem~\ref{thm:stv_expression} formula on this subset.
\end{remark}

We now define the previously proposed \textbf{Shapley interaction index} (\cite{Grabisch99})
\footnote{
The Shapley interaction index can also defined by a random order process. To compute the interaction index for a set $S$, the set $S$ is fused into an artificial player. The players are ordered randomly as before and the discrete derivative $\delta_S$ is evaluated at the set of players which precede the artificial player in the ordering. The combinatorial weights that occur in Equation~\ref{eq:siv} arise from this random process.}. The Shapley interaction index for a subset $S \subseteq N$ for a function $\V \subseteq \mathcal{G}^N$ is defined as:
    \begin{equation}
    \label{eq:siv}
        I_{Sh}(\V, S) := \sum_{T \subseteq N \setminus S} \frac{(n-t-s)! t!}{(n-s+1)!} \delta_S \V(T)
    \end{equation}

\comment{
 \begin{remark}
 Let us briefly discuss why we inserted 'Taylor' in the title of the method. First, there is an obvious analogy between truncated Taylor series expansions and the \stv. Just as the Taylor series approximates a function at a point using its derivatives at another point and the Lagrangian remainder, \stv approximates the function at the input using its discrete derivatives at the empty set and interaction terms analogous to the Lagrangian remainder that model interactions of order $k$ and higher. In fact, we can show that the Shapley-Taylor indices result from applying the Taylor series with the standard Lagrangian remainder to the \emph{multilinear extension} of the discrete function. Multilinear extensions, first studied by ~\cite{Owen72}, turn a game (discrete function) into a continuous one by making every player (feature) $i$ join the game with some probability $x_i$.
 \end{remark}
}

%% file: axioms.tex
\section{Axiomatization of the Shapley-Taylor Interaction Index}
\label{sec:axiom}

In this section, we axiomatize \stv.
Let $\mathcal{G}^{N}$ denote the set of \games  on $N$ \players. 

The first three axioms are all variants of the standard Shapley axioms generalized to interactions by~\cite{Grabisch99}; they were used in the axiomatization of \siv. 

\begin{enumerate}
    \item \textbf{Linearity axiom:} $\interaction(\cdot)$ is a linear function; i.e. for two \games $\V_1, \V_2 \in \mathcal{G}^N$, $\interaction_{S}(\V_1 + \V_2) = \interaction_{S}(\V_1) + \interaction_{S}(\V_2)$ and $\interaction_{S}(c \cdot \V_1) = c \cdot \interaction_{S}(\V_1)$.
    \item \textbf{Dummy axiom:} If $i$ is a dummy \player for $\V$, i.e. $\V(S) = \V(S \setminus i) + \V(i)$ for any $S \subseteq N$ with $i \in S$, then 
    \begin{enumerate}[(i)]
        \item $\interaction_i(\V) = \V(i)$
        \item for every $S \subseteq N$ with $i \in S$, we have $\interaction_{S}(\V) = 0$
    \end{enumerate}
    \item \textbf{Symmetry axiom:} for all \games $\V \in \mathcal{G}^N$, for all permutations $\pi$ on $N$, :
    $$ \interaction_{S}(\V) = \interaction_{\pi S}(\pi \V)$$
    where $\pi S := \{ \pi(i) | i \in S \}$ and the \game $\pi v$ is defined by $(\pi \V)(\pi S) = \V(S)$, i.e. it arises from
    relabeling of \players $1,\ldots, n$ with the labels $\pi(1), \ldots , \pi(n)$.
\end{enumerate}

In addition to these three axioms from~\cite{Grabisch99}, we use the \efficiency axiom. 
Again, this is a generalization of the standard \efficiency axiom for Shapley values.

\begin{enumerate}
\setcounter{enumi}{3}
\item  \textbf{Efficiency axiom:} for all \games $\V$, $\sum_{S \subseteq N, |S| \leq k} \interaction_S(\V) = \V(N) - \V(\emptyset)$
\end{enumerate}

Finally, we introduce the \emph{Interaction Distribution} axiom. This axiom is defined for \games that we call \emph{interaction \games}. An interaction \game parameterized by a set $T$, has the form $\V_T(S)= 0$ if $T \not\subseteq S$ and has a constant value $\V_T(S)=c$ when $T \subseteq S$. These \games model pure interaction among the members of $T$ (when $|T| >1$---the combined presence of the members of $T$ is necessary and sufficient for the \game to have non-zero value. We call $|T|$ the the order of the interaction of the interaction \game. In machine learning terminology, the function $\V_T(S)$ is a model with a single feature cross of all the (categorical) features in $T$, with a coefficient of $c$ for this cross.

\begin{enumerate}
\setcounter{enumi}{4}
\item  \textbf{Interaction Distribution Axiom:} 
For an interaction \game $\V_T$ parameterized by the set $T$, for all $S$ with 
$S \subsetneq T$ and $|S| < k$, where $k$ denotes the order of explanation, we have $\interaction_{S}(\V_T) = 0$.
\end{enumerate}

This intention of the axiom is to ensure that higher order interactions are \emph{not} attributed to lower order terms or alternatively, that interaction terms of order (say) $l < k$ capture interactions of size $l$ only. The axiom relaxes this restriction for terms of size exactly $k$, which accumulates the contributions of the higher order interactions.This is reductive in the sense that the $k$th order terms don't just reflect the $k$th order interactions, but also the higher order ones. However, as $k$ is increased, the reductivity decreases at the cost of a longer explanation.

We are now ready to state our main result:

\begin{theorem}\label{thm:uniqueness}
\stv are the only interaction indices that satisfy axioms 1-5.
\end{theorem}
\begin{proof}
We prove this in two steps. First we show that \stv satisfy the axioms (Propositions 1--3). Next we show that
any method that satisfies the axioms assigns specific interaction values to unanimity \games (Proposition 4). 

\begin{proposition}
\stv satisfy the Linearity, Dummy and Symmetry axioms
\end{proposition}
\begin{proof}
In Equation~\ref{eq:permutation}, \stv are defined as expected values of certain discrete derivatives $\delta_S$. The discrete derivative satisfies linearity conditions. Hence \stv satisfy the linearity axiom. The symmetry axiom follows from the fact that \stv are defined as expectations over all permutations. To show the dummy axiom, we note that the discrete derivative $\delta_S \V(T)$ can be rewritten as $\delta_{S \setminus i} \V(T \cup i) - \delta_{S \setminus i} \V(T)$ for any $i \in S$. If $i$ is a dummy \player, it follows that $\delta_S \V(T) = 0$. Consequently $\interaction_S (v) = 0$. Furthermore, $\interaction_{\{ i \}}(\V) = \delta_i \V(\emptyset) = \V(i)$. Thus \stv also satisfy the dummy axiom.
\end{proof}

\begin{proposition}
\stv satisfy the Interaction Distribution axiom.
\end{proposition}

\begin{proof}
Consider an interaction \game $\V_T$ and $S$ such that $|S| < |T|$. Notice that $\interaction_S(\V_T) = \delta_S(\V_T) = \sum_{W \subseteq S} (-1)^{w-s} \V_T(W)$. Since $|W| < |T|$, we know that $\V_T(W) = 0$. Hence $\interaction_S(\V_T) = 0$.
This shows that \stv satisfies the Interaction Distribution axiom.
\end{proof}

\begin{proposition}\label{prop:efficiency}
\stv satisfy the efficiency axiom. Formally,
$$\sum_{S, |S| \leq k} \interaction_{S}(\V) = \V(N) - \V(\emptyset)
$$
\end{proposition}

\begin{proof}
We prove the proposition for unanimity functions. By using additivity axiom, this extends to all \games.
Consider the unanimity \game $u_T$ defined as $u_T(S) = 1$ if $S \supseteq T$, and $0$ otherwise.

\begin{lemma}
For all sets $W$, $\delta_S u_T(W) = 0$ if either $S \not\subseteq T$ or $T \not\subseteq S \cup W$.
\end{lemma}

Fix a permutation $\pi$ as the ordering of the \players. Let $k$ be the order of explanation. The two lemmas above have following implications:\\
- $\interaction_{S, \pi}(u_T) = 0$ if $|S| > |T|$. \\
- For $|T| < k$, $\interaction_{S, \pi}(u_T) = 1$ iff $S = T$, otherwise $0$. \\
- For $|T| < k$, this proves efficiency.
\\
The remaining case is $|T| \geq k$. For $|S| < k$, we have $\interaction_{S, \pi}(u_T) = \delta_S u_T(\emptyset)$. This is $0$ from second lemma.

Finally, consider $|S| = k$. Here, we claim that $\interaction_{S, \pi}(u_T) = 1$ iff $S$ is same as the last $k$ elements of T in the permutation order $\pi$; otherwise it is $0$.

Let $T_k := \{ \text{last $k$ elements of $T$} \} \}$. Notice that if $S = T_k$, then $T \subseteq \pi^S \cup S$. Therefore, $u_T(S \cup \pi^S) = 1$, but $u_T(S' \cup \pi^S) = 0$ for any subset $S' \subset S$. Hence $\delta_S u_T(\pi^S) = 1$.

To prove the other side, there are two cases:

(a) $S$ has an element that is not in $T$: first lemma gives $\interaction_{S, \pi}(u_T) = 0$.\\
(b) $S$ does not contain an element of $T_k$: then $T \not\subseteq S \cup \pi^S$. Second lemma implies $\interaction_{S, \pi}(u_T) = 0$.

This proves the efficiency axiom for $|T| \geq k$.
\end{proof}

To prove the uniqueness, we investigate interactions for unanimity \games.
Recall that, a unanimity \game parameterized by a set $T \subseteq N$ is defined as $\V_T(S) = 1$ iff $S \supseteq T$, and $0$ otherwise. Thus unanimity \games are a subset of Interaction \games used for the Interaction Distribution Axiom.

Since the family of unanimity \games $\{ \V_T \}_{T \subseteq N, T \neq \emptyset}$ forms a linear basis of $\mathcal{G}^N$, it is sufficient to show that  any interaction index that satisfies the axioms 1-5 
assigns specific values on \emph{unanimity \games}. This is shown in the next proposition.

\begin{proposition}\label{prop:unanimity_phi}
Consider the unanimity \game $\V_T$ defined as $\V_T(S) = 1$ if $T \subseteq S$, otherwise $0$. 
Let $k$ be the order of explanation. Let $\phi$ be an interaction index that satisfies the axioms 1-5, then 
\begin{equation}\label{eq:unanimity_stv}
\phi_S(\V_T) = 
\begin{cases}
1 & \text {if $S = T$ and $|S| < k$} \\
0 & \text{if $S \neq T$ and $|S| < k$} \\
\frac{1}{{t \choose k}} & \text{if $S \subseteq T$ and $|S| = k$} \\
0 & \text{$|S| = k$, but $S \not\subseteq T$ } \\
\end{cases}
\end{equation}
\end{proposition}

\begin{proof}
Consider a unanimity \game $\V_T$ and an interaction index $\phi_S$ that satisfies axioms 1-5.
We want to derive $\phi_S(\V_T)$ using the axioms.

Start with the dummy axiom. If $i \not\in T$, then $i$ is a dummy \player for $\V_T$. This implies that if $S \setminus T \neq \emptyset$, then $\phi_S(\V_T) = 0$.

Consider $|S| < k$. The interaction distribution axiom states that $\phi_S(\V_T) = 0$ if 
$S  \subsetneq T$. Next we use the dummy axiom. Note that, for all $j \not\in T$,
$\delta_j \V_T(S) = 0$. Thus $j$ is dummy \player. The dummy axiom implies that 
$\phi_S(\V_T) = 0$ if $j \in S$. Hence, $\phi_S(\V_T) = 0$ for all $S$ such that 
$|S \setminus T| > 0$.

Using the efficiency axiom, we have $\sum_S \phi_S(\V_T) = \V_T(N) = 1$. As we saw before, $\phi_S(\V_T) = 0$ if $S \neq T$.
Hence the sum reduces to  $\phi_S(\V_S) = 1$, when $S = T$.

Finally, consider the case where size of $S$ is same as the order of explanation, i.e.  $|S| = k$.
As we saw earlier, $\phi_S(\V_T) = 0$ if $S \not\subseteq T$. Using the efficiency axiom:
$\sum_{S \subseteq T}  \phi_{S}(\V_T)  = 1$. Furthermore, since $k$ is the order of explanation, the interaction index is
defined to be $0$ for sets larger than $k$. Hence, $\sum_{S \subseteq T, |S| = k}  \phi_{S}(\V_T)  = 1$.

The symmetry axiom implies that each of these terms must be equal. Hence $\phi_{S}(\V_T) = \frac{1}{{t \choose k}}$.
\end{proof}

Proposition~\ref{prop:unanimity_phi} shows that the interaction values for a unanimity \game $\V_T$ depend only on $|T|$ and the order of explanation $k$. Since unanimity \games form a basis of $\mathcal{G}^N$, the interaction values extend in a unique way to all the \games. Thus there is a unique method (\stv) that satisfies the axioms 1-5.
\end{proof}

\subsection{Connection to Taylor series}\label{sec:taylor}
In this section, we connect \stv to the Taylor series of the multilinear extension.
%We show a one-to-one correspondence between the terms in the Taylor series of the multilinear extension of a \game and the \stv of the \game.

%\emph{Multilinear extension} of a \game was first introduced by Owen~\cite{Owen72}. He showed that there is a unique multilinear function that coincides with $v$ on the corners of the cube $\{0, 1\}^N$. 
The multilinear extension of $\V$ is defined as follows:

\begin{equation}\label{eq:multilinear}
f(x) = \sum_{S \subseteq N} \V(S) \prod_{i \in S} x_i \prod_{i \not\in S} (1 - x_i)     
\end{equation}

where $x_i \in [0, 1], i \in N$. The multilinear extension has a probabilistic interpretation: if \player $i$ is set with probability $x_i$, then $f(x)$ denotes the expected value of the \game. Let 
$g(t) = f(t, t, \ldots, t)$ for $t \in [0, 1]$. Then $g(0) = \V(\emptyset)$ and $g(1) = \V(N)$. For a set $S \subseteq N$, with $|S| = j$, define 
$$\Delta_S f(x) := \frac{\partial^{j} f}{\partial x_{i_1} \ldots \partial x_{i_j}} \text{~~~where~~} S = \{i_1, \ldots, i_j \}$$
Consider the $(k-1)^{\text{th}}$ order Taylor expansion of $g(\cdot)$ at $0$ with Lagrange remainder and the corresponding multivariate expansion of each term in terms of $\Delta_S f(\cdot)$:

% TODO(kedar): does the following equation make sense?

\begin{align}
 & \V(N) - \V(\emptyset) = g(1) - g(0)  \notag \\
 & = \frac{g'(0)}{1 !} + \ldots +  \frac{g^{(k-1)}(0)}{(k-1)!} + \int_{t=0}^1 \frac{(1-t)^{k-1}}{(k-1)!} g^{(k)}(t)  dt \label{eq:taylor} \\
 & = \begin{multlined}[t]
 \sum_{j \leq k-1} \sum_{|S| = j} \Delta_S f(0, \ldots, 0) + \notag \\ 
  \sum_{|S| = k} \int_{t=0}^1 k (1-t)^{k-1} \Delta_S f(t, \ldots, t) dt 
 \end{multlined} 
\end{align}

\begin{theorem}
\label{thm:taylorterm}
Let $k$ be the order of explanation and $j < k$. Then $j^{\text{th}}$ order \stv can be obtained from the $j^{\text{th}}$ order terms in the Taylor series. 
$$\Delta_S f(0, \ldots, 0) = \interaction_{S}(\V) \text{~~~where~~} |S| < k$$ 
The $k^{\text{th}}$ order \stv can be obtained from the Lagrange remainder term:
$$ \int_{t=0}^1 k (1-t)^{k-1} \Delta_S f(t, \ldots, t) dt = \interaction_{S}(\V) \text{~~where~} |S| = k$$
\end{theorem}
We give the proof of the Theorem and Equation~\ref{eq:taylor} in the appendix. For $k = 1$, we see that the definition of \stv in Equation~\ref{eq:permutation} is exactly the Shapley value. Furthermore, Theorem~\ref{thm:taylorterm} reduces to the result by \cite{Owen72} (Theorem~5) showing that the Shapley value can be obtained by integrating the gradient of $f$ along the diagonal line.

\section{Comparison of Shapley Interaction Index and Shapley Taylor Index}
\label{sec:example}

\subsection{A Linear Model with Crosses}
\label{sec:linear}

We illustrate the difference between \stv and \siv using a function that models linear regression with crosses. Consider the following \game with 3 binary \players: $F(x_1, x_2, x_3) := [x_1] + [x_2] + [x_3] + c [x_1]*[x_2]*[x_3]$; $[x_i]$ is $1$ if feature $x_i$ is present in the input, and zero otherwise.  The last term of the function models a cross between the three features. This is a very simple function; we would expect the main effects to be $1$ each, and the interaction effect of $c$ to be divided equally among the three pairs. 

Indeed, this is what we see with the \stv. The pairwise interaction terms for each of the three pairs is $c/3$, and the total interaction effect matches the coefficient of the cross term, $c$. The main effect terms for \stv are $1$ each, as expected. 

In contrast, those for \siv are $c/2$ each, and the total interaction effect is $3c/2$. Whether $c$ is negative or positive, \siv \emph{amplifies} the magnitude of interaction effect \footnote{If the function does not have crosses of size more than $2$, the two methods coincide in the interaction terms.}. While \siv does not directly define main effect terms, a natural way to define them is to subtract the interaction terms from the Shapley values: $\Phi_{i,i} = \Phi_i - \frac{1}{2}\sum_{i \neq j} \Phi_{i,j}$, where $\Phi_i$ is the $i$th Shapley value and $\Phi_{i,j}$ is the Shapley Interaction index for $\{ i, j \}$. With this definition, the main effect terms and interaction terms satisfy the efficiency axiom.

Using this definition, the main effect terms for \siv are $1-\frac{c}{3}$ each. When $c$ is larger than $3$, the main effect term is negative, which does not match our expectation about the function.

In general, we observe that \siv returns inflated interaction values, possibly because it was not designed to satisfy efficiency, and a consequence of this is that the main effect terms, computed by subtracting the inflated interaction terms from the Shapley value, can have incorrect signs. 

In general, there is no simple way to account for the amount of inflation in \siv. The following example shows that the size of the inflation varies with the size of the cross-terms present in the model: For example, for $F(x_1, \ldots, x_n) := [x_1] * \ldots * [x_n]$, the total interaction effect in \stv is $1$, i.e., $1/{n \choose 2}$ for each pair), while that for \siv is $n/2$ (i.e. $\frac{1}{n-1}$ for each pair). The inflation factor $n/2$ depends on the size of the interaction.

\subsection{The Majority Function}
\label{sec:majority}
Next, we consider the majority \game $\V_{maj}$ over the set of \players $N$ is defined by $\V_{maj}(S) = 1$ if $|S| \geq n/2$ and is $0$ otherwise; here $n=|N|$.

It is easy to check that the singleton terms for pairwise \stv for a majority \games are uniformly zero (for $n >2$). The
pairwise terms are each equal to reciprocal of the number of pairs, $2/(n \cdot (n-1))$, a simple consequence of symmetry and
efficiency. The analytical conclusion from this is that majority \games are all about interaction, and this is intuitively reasonable.

In contrast, the pairwise \siv are uniformly zero for every pair of \players! This is unexpected; one would imagine that in a three \player majority \game there would be non-zero interaction values for sets of size $2$, the size at which majority arises. The results from studying pairwise \siv seems to suggest that there is \emph{no} pairwise interaction! However, \siv of larger interaction sizes do have non-zero values. For instance, for a majority \game with $3$ \players, the interaction value corresponding to all three \players is $-2$, a non-intuitive number. The pattern becomes even more confusing for majority \games of more \players.

Figure~\ref{fig:majority} (page 8) shows the sum of the \siv (in log scale) for all subsets of \players for majority \games as number of \players increases. This displays two non-intuitive characteristics. First, the total interaction diverges (recall the plot is in log scale) despite the fact that the \game value being shared among the \players is constant ($1$). Second, the \emph{sign} of the total interaction alternates. In fact, for a \game of a fixed size, \emph{every} non-singleton interaction value has the same sign. (So the sign alternation is not due to cancellation across the interactions.) There is no intuitive reason for this sign alternation. 

This discussion highlights the value of the \stv interaction distribution approach over that of \siv. To study interactions using \siv, one has to compare $2^n$ quantities, and even then the results appear unstable. In contrast, studying $O(n^2)$ quantities using the pairwise \stv gives us some insight into the functions behavior, no matter what the size of the set of interacting elements.

%% file: evaluation.tex
\section{Applications}
\label{sec:applications}

We study three tasks. The first scores sentences of text for sentiment. The model is a convolutional neural network from~\cite{sentiment} trained over movie reviews\footnote{To ablate a word, we zero out the embedding for that word.}. The second is a random forest regression model built to predict house prices. We use the Boston house price dataset (\cite{boston_housing}). The dataset has $12$ numerical features and one binary feature.\footnote{There are $506$ data points. We split the data into $385$ training and $121$ test examples. We used \texttt{scikit-learn} to train a random forest model. When we ablate a feature, we replace it by its training-data mean.} The third model, called QANet~\cite{YuDLLZC2018}, solves reading comprehension, i.e., identifying a span from a context paragraph as an answer to a question; it uses the SQuAD dataset~(\cite{RajpurkarZLL16}). \footnote{We use Equation~\ref{eq:stv_formula} to compute the \stv. However, the computational efficiency is $O(2^n)$, where $n$ is the number of features. For the sentiment analysis and question answering tasks, $n$ is large enough for this to be prohibitive. So we first use an attribution technique (Integrated Gradients~\cite{STY17}) to identify a small subset of influential features and then study interactions among that subset.}

\subsection{Insights}

  For the sentiment task, we ran \stv across a hundred test examples and identified word pairs with large interactions. Table~\ref{tab:sentiment_interactions} shows some of the interesting interactions we found. 

\begin{table}[!htb]
    \centering
    \small
    \begin{tabular}{p{6cm}|p{1.8cm}}
    \hline
    Example sentence (interaction between bolded words) & interaction effect\\
    \hline
%    It's \textbf{not} life-affirming -- its \textbf{vulgar} and mean, but I liked it  & 0.509 \\
    Aficionados of the whodunit wo\textbf{n't} be \textbf{disappointed} & 0.778 \\
    Watching these eccentrics is \textbf{both} \textbf{inspiring} and pure joy & 0.1 \\
    \textbf{A} \textbf{crisp} psychological drama (and) a fascinating little thriller & 0.48 \\
    With three \textbf{excellent} principal singers, a \textbf{youthful} and good-looking diva \ldots & -0.224 \\
    Australian actor/director \textbf{John} Polson and $\ldots$ make a \textbf{terrific} effort $\ldots$ & -0.545  \\
    \hline
    \end{tabular}
    \caption{Examples of different types of interactions. The interaction effect is the fraction of total change.}
    \label{tab:sentiment_interactions}
\vspace{-4mm} %reduce too much white space
\end{table}

The first example captures negation. The second example shows intensification; the effect of 'inspiring' is amplified by the 'both' that precedes it. The third example demonstrates a kind of intensification that would not typically appear in grammar texts; 'crisp' is intensified because it appears at the beginning of the review, which is captured by the interaction with 'A', whose main effect is nearly zero. The fourth examples shows complementarity; the interaction effect is of opposite sign to the main effects. The final example shows that sentiment expressed in third person is suppressed. This is natural because reviews are first-person artifacts.

For the random forest regression task, we found that most of the influential pairwise feature interactions are substitutes; 
e.g. when the predicted price is lower than the average price, the main effects are negative, but the pair-wise interaction
values are positive. We think this is because the features are correlated. We show a plot of main effects vs interaction effects of pairs of features in Figure~\ref{fig:substitutions}.

\begin{figure}[htb]
 \begin{minipage}{0.5\textwidth}
 \centering
  \includegraphics[width=0.9\linewidth]{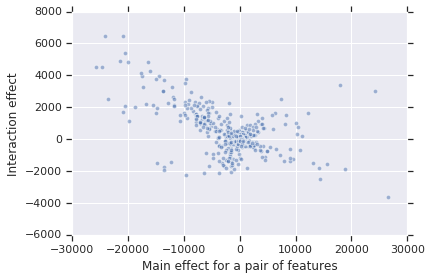}
  
  \caption{{\small A plot of main effects of pairs of features vs interaction effect for the pair. Negative slope indicates substitutes.}}
    \label{fig:substitutions}
  \end{minipage}    
   \begin{minipage}{0.5\textwidth}
   \centering
  \includegraphics[width=0.9\linewidth]{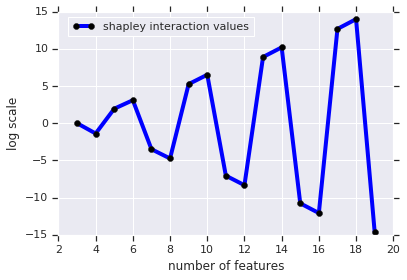}
\caption{{\small The sum of \siv (in log scale) for all subsets of \players for majority \games as a function of the number of \players.}}
\label{fig:majority}  
  \end{minipage}      
    
%\vspace{-4mm} %reduce too much white space
  \end{figure}

For the reading comprehension task, previous analyses (cf.~\cite{MudrakartaTSD18}) focused on whether the model ignored important question words. These analyses did not identify which paragraph words did the important question words match with. Our analysis identifies this. See the examples in Figure~\ref{fig:qanet_matches}.

\begin{figure}[!htb]
  \includegraphics[width=0.45\linewidth]{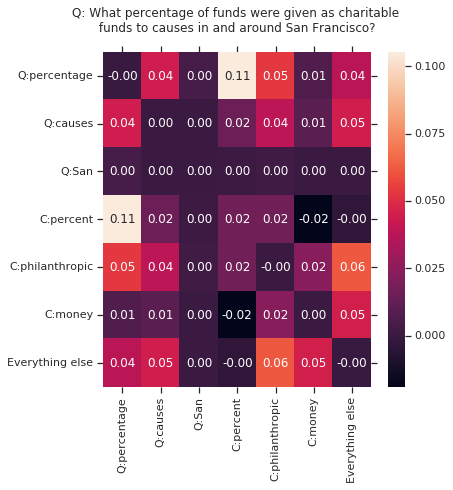}
   \includegraphics[width=0.45\linewidth]{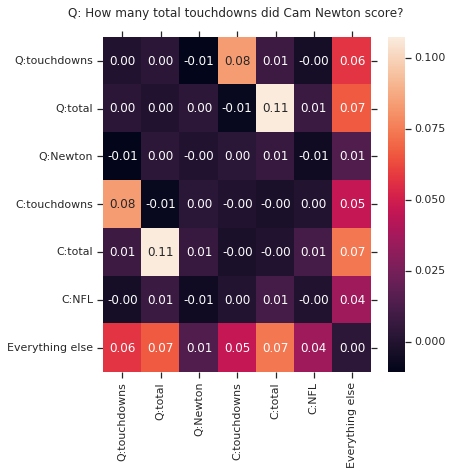}
  \caption{{\small First example shows a match between \emph{percentage} in question to \emph{percent} in the context. Second example shows \emph{total touchdowns} matching between the question and the context. }}
    \label{fig:qanet_matches}
\vspace{-4mm} %reduce too much white space
  \end{figure}

%% file: appendix.tex
\clearpage
\section{Appendix}

\subsection{Proof of Theorem~\ref{thm:stv_expression}}
\label{sec:stv_expression}
\begin{proof}
For a set $T \subseteq N$, define the \emph{M\"obius coefficients} as:
\[
a(T) = \sum_{S \subseteq T} (-1)^{t-s} \V(S), ~~~~~~ T \subseteq N
\]
Decomposition of any function $\V$ into unanimity \games can be written in terms of $a(T)$'s as follows:
\begin{equation*}
 \V(S)   = \sum_{T \subseteq N} a(T) u_T(S) 
\end{equation*}
Using the linearity axiom, we can extend $\interaction_S$ from unanimity \games to $\V$ as follows:
\begin{equation*}
 \interaction_{S}(v)  = \sum_{T \subseteq N} a(T) \interaction_{S}(u_T)   
\end{equation*}  
 Using the \prefixvalues for unanimity \games from Proposition~\ref{prop:unanimity_phi}, we get:
\begin{align*}
\interaction_{S}(v) & = \sum_{T \subseteq N, S \subseteq T} a(T) \frac{1}{{t \choose k}}   \label{eq:stv_at} \\
 & = \sum_{W \subseteq N \setminus S} a(W \cup S) \frac{1}{{w+k \choose k}}  \text{~ ~ ~ where $W = T \setminus S$} \notag \\
 & = 
 \begin{multlined}[t]
 \sum_{W \subseteq N \setminus S}  \sum_{U \subseteq W} (-1)^{u-w} \delta_S \V(U)  \frac{1}{{w+k \choose k}}  \\
    \text{(using Lemma~\ref{lemma:mobius} below)} \notag
 \end{multlined} \\
 & = \sum_{U \subseteq N \setminus S} \delta_S \V(U) \sum_{W \supseteq U, W \subseteq N \setminus S} \frac{(-1)^{u-w}}{{w+k \choose k}} \notag \\
\end{align*}

Now we analyze the inner sum. We use the following identity: ${w+k \choose k} = 1/(k \cdot B(w+1, k))$, where $B()$ is the Beta function.

\begin{align*}
& \sum_{U \subseteq W \subseteq N \setminus S} \frac{(-1)^{u-w}}{{w+k \choose k}}  \\ 
& =  \sum_{W \supseteq U, W \subseteq N \setminus S} (-1)^{u-w} k \cdot B(w+1, k)  \\
& = k \sum_{w=u}^{n-k} {n-k-u \choose w-u} (-1)^{u-w} B(w+1, k)  \\
& = k \sum_{w=u}^{n-k} {n-k-u \choose w-u} (-1)^{u-w} \int_{0}^1 x^w (1-x)^{k-1} dx  \\
& = \begin{multlined}[t]
k \int_{0}^1 \sum_{w=u}^{n-k} {n-k-u \choose w-u} (-1)^{u-w} x^w (1-x)^{k-1} dx \\
\text{(exchanging sum \& integral)}
\end{multlined} \\
& = \begin{multlined}[t]
k \int_{0}^1 x^u (1-x)^{k-1} \sum_{w'}^{n-k-u} {n-k-u \choose w'} (-1)^{w'} x^{w'} dx \\
\text{(setting $w' = w-u$)} 
\end{multlined} \\
& = k \int_{0}^1 x^u (1-x)^{k-1} (1-x)^{n-k-u} dx \\
& = k B(u+1, n-u)  \text{~ ~ ~ ~ (definition of Beta function)} \\
& = \frac{k}{n} \frac{1}{{n-1 \choose u}}
\end{align*}

We use this expression for the inner sum in the above equation to get:
$$\interaction_{S}(v) = \frac{k}{n} \sum_{U \subseteq N \setminus S} \delta_S \V(U)  \frac{1}{{n-1 \choose u}} $$
This finishes the proof.
\end{proof}

The next Lemma provides a relation between the M\"obius coefficients $a(T)$ and the discrete derivatives.
\begin{lemma}\label{lemma:mobius}
M\"obius coefficients and discrete derivatives are related by following relation:
$$a(T \cup S) = \sum_{W \subseteq T} (-1)^{t-w} \delta_S \V(W)$$
for $S$ and $T$ such that $S \cap T = \emptyset$.
\end{lemma}

\begin{proof}
%a(T) = \sum_{W \subseteq T} (-1)^{t-s} \V(W)$$
Let $S$ and $T$ be two sets such that $S \cap T = \emptyset$. We have
\begin{align*}
    a(T \cup S) 
& = \sum_{W' \subseteq T \cup S} (-1)^{t+s-w'} \V(W') \\
& = \begin{multlined}[t]  
\sum_{W \subseteq T} \sum_{U \subseteq S} (-1)^{t+s-w-u} \V(U \cup W) \\
\text{(where $W = W' \cap T$ and $U = W' \cap S$)} 
\end{multlined} \\
& = \sum_{W \subseteq T} (-1)^{t-w} \sum_{U \subseteq S} (-1)^{s-u} \V(U \cup W) & \\
& = \sum_{W \subseteq T} (-1)^{t-w} \delta_S \V(W) \\
\end{align*}
\end{proof}

\subsection{Proof of Theorem~\ref{thm:taylorterm}}
\begin{proof}
Recall that $g(t) = f(t, \ldots, t)$ for $t \in [0, 1]$.
First, we derive the multivariate expression in Equation~\ref{eq:taylor}.
$$
\frac{g^{(j)}(0)}{j !} = \sum_{S \subseteq N, |S| = j} \Delta_S f(0, \ldots, 0)
$$

Consider the expansion of $g^{(j)}(0)$:
\begin{align*}
\frac{g^{(j)}(0)}{j !} & = \frac{1}{j!} \sum_{i_1} \sum_{i_2} \cdots \sum_{i_j} \frac{\partial^{j} f(0, \ldots, 0)}{\partial x_{i_1} \ldots \partial x_{i_j}} \\
\end{align*}
Notice that $f$ is a multilinear function. Therefore, 
 only the mixed partial terms survive. Furthermore, all $j!$ mixed partials wrt $x_{i_1}, \ldots, x_{i_j}$ are identical. Hence, we can simplify the above equation to:

\begin{align*}
\frac{g^{(j)}(0)}{j !} & = \sum_{i_1 < i_2 < \ldots < i_j} \frac{\partial^{j} f (0, \ldots, 0)}{\partial x_{i_1} \ldots \partial x_{i_j}} \\
& = \sum_{S \subseteq N, |S| = j} \Delta_S f(0, \ldots, 0) \\
\end{align*}

Similarly, consider the multivariate Lagrange remainder term in Equation~\ref{eq:taylor}:
$$ \int_{t=0}^1 \frac{(1-t)^{k-1}}{(k-1)!}  g^{(k)}(t) dt $$
As before, in the derivative of $g$, only the mixed partial terms are left:

\begin{equation*}
\frac{g^{(k)}(t)}{k !} = \sum_{S \subseteq N, |S| = k} \Delta_S f(t, \ldots, t)
\end{equation*}
We use this expression in Lagrange remainder term and interchanging the order of integral and summation. Note that there is an extra factor of $k$ that survives on the right side.
\begin{align}\label{eq:lagrangesum}
\int_{t=0}^1 \frac{(1-t)^{k-1}}{(k-1)!}  g^{(k)}(t) dt  =  \notag \\
 \sum_{S \subseteq N, |S| = k}  \int_{t=0}^1 k (1-t)^{k-1} \Delta_S f(t, \ldots, t) dt
\end{align}
% TODO(kedar): what does the next sentence mean?
%This shows the equivalence of Equation~\ref{eq:taylor} and Equation~\ref{eq:taylor}.

For the rest of the proof of the theorem, we consider the special case unanimity \games $u_W$ for a set $W \subseteq N$ and the corresponding multilinear extension:
\begin{equation}\label{eq:unanimity-multilinear}
f_W(x) = \prod_{i \in W} x_i     
\end{equation}

We prove the theorem for the unanimity \games. Since unanimity \games form the basis, the general case follows from linearity axiom.

Recall that a set $S \subseteq N$, $|S| = j$:
$$\Delta_S f_W(x) := \frac{\partial^{j} f_W}{\partial x_{i_1} \ldots \partial x_{i_j}} \text{~~~where~~} S = \{i_1, \ldots, i_j \}$$ 
Using Equation~\ref{eq:unanimity-multilinear}, we get 
\begin{align*}
\Delta_S f_W(x) & = \prod_{i \in W \setminus S} x_i \text{~~~iff $S \subseteq W$}     \\
& = 1 \text{~~~when $S = W$ and $0$ otherwise.}
\end{align*}
Hence $\Delta_S f_W(0) = 1$ iff $S = W$ and $0$ otherwise. This gives us:
\begin{equation}\label{eq:taylorterm1}
\Delta_S f_W(0, \ldots, 0) = \interaction_S (u_W)    
\end{equation}

This proves the result for $j < k$.

Next we analyze the Lagrange Remainder term. Consider a set $S \subseteq N, |S| = k$. We use Eqaution~\ref{eq:unanimity-multilinear} to evaluate $\Delta_S f_W(t, \ldots, t)$:

\begin{align}
& \int_{t=0}^1 k (1-t)^{k-1} \Delta_S f_W(t, \ldots, t) dt \\
& = \int_{t=0}^1 k (1-t)^{k-1}  \left( \prod_{i \in W \setminus S} t \right) dt & \text{($\Delta_S f_W(\cdot) = 0$ if $S \not\subseteq W$)} \notag \\
& = \int_{t=0}^1 k (1-t)^{k-1}  \left(t^{w-k} \right) dt & \text{$|W| = w$ and $|S| = k$} \notag \\
& =  k \int_{t=0}^1 t^{w-k} (1-t)^{k-1} dt & \notag \\
& = k \cdot B(w+1, k) & \text{$B(\cdot, \cdot)$ is the Beta function}\notag \\
& = 1 / {w \choose k}  & \notag \\
& = \interaction_S(u_W) & \text{from Equation~\ref{eq:unanimity_stv}}\notag
\end{align}

This finishes the proof for the remainder term.
\end{proof}

%% file: main.bbl
\begin{thebibliography}{10}

\bibitem{AS74}
{\sc Aumann, R.~J., and Shapley, L.~S.}
\newblock {\em Values of Non-Atomic Games}.
\newblock Princeton University Press, Princeton, NJ, 1974.

\bibitem{BSHKHM10}
{\sc Baehrens, D., Schroeter, T., Harmeling, S., Kawanabe, M., Hansen, K., and
  M{\"u}ller, K.-R.}
\newblock How to explain individual classification decisions.
\newblock {\em Journal of Machine Learning Research 11\/} (2009), 1803--1831.

\bibitem{BMBMS16}
{\sc Binder, A., Montavon, G., Lapuschkin, S., M{\"u}ller, K.-R., and Samek,
  W.}
\newblock Layer-wise relevance propagation for neural networks with local
  renormalization layers.
\newblock In {\em ICANN\/} (2016).

\bibitem{Cohen2005}
{\sc Cohen, S., Ruppin, E., and Dror, G.}
\newblock Feature selection based on the shapley value.
\newblock In {\em Proceedings of the 19th International Joint Conference on
  Artificial Intelligence\/} (San Francisco, CA, USA, 2005), IJCAI’05, Morgan
  Kaufmann Publishers Inc., p.~665–670.

\bibitem{Cui}
{\sc Cui, T., Marttinen, P., and Kaski, S.}
\newblock Recovering pairwise interactions using neural networks.
\newblock {\em CoRR abs/1901.08361\/} (2019).

\bibitem{Datta2015Influence}
{\sc Datta, A., Datta, A., Procaccia, A.~D., and Zick, Y.}
\newblock Influence in classification via cooperative game theory.
\newblock In {\em Proceedings of the 24th International Conference on
  Artificial Intelligence\/} (2015), IJCAI’15, AAAI Press, p.~511–517.

\bibitem{Datta2016Algorithmic}
{\sc Datta, A., Sen, S., and Zick, Y.}
\newblock Algorithmic transparency via quantitative input influence: Theory and
  experiments with learning systems.
\newblock {\em 2016 IEEE Symposium on Security and Privacy (SP)\/} (2016),
  598--617.

\bibitem{ANOVA}
{\sc Gelman, A., et~al.}
\newblock Analysis of variance—why it is more important than ever.
\newblock {\em The annals of statistics 33}, 1 (2005), 1--53.

\bibitem{Grabisch99}
{\sc Grabisch, M., and Roubens, M.}
\newblock An axiomatic approach to the concept of interaction among players in
  cooperative games.
\newblock {\em International Journal of Game Theory 28\/} (11 1999), 547--565.

\bibitem{boston_housing}
{\sc Harrison, D., and Rubinfeld, D.~L.}
\newblock Hedonic prices and the demand for clean air.
\newblock {\em J. Environ. Economics \& Management 5\/} (1978), 81--102.

\bibitem{sentiment}
{\sc Kim, Y.}
\newblock Convolutional neural networks for sentence classification.
\newblock In {\em Proceedings of the 2014 Conference on Empirical Methods in
  Natural Language Processing, {EMNLP} 2014, October 25-29, 2014, Doha, Qatar,
  {A} meeting of SIGDAT, a Special Interest Group of the {ACL}\/} (2014),
  A.~Moschitti, B.~Pang, and W.~Daelemans, Eds., {ACL}, pp.~1746--1751.

\bibitem{Lundberg2017AUA}
{\sc Lundberg, S., and Lee, S.-I.}
\newblock A unified approach to interpreting model predictions.
\newblock In {\em NIPS\/} (2017).

\bibitem{Lundberg18}
{\sc Lundberg, S.~M., Erion, G.~G., and Lee, S.}
\newblock Consistent individualized feature attribution for tree ensembles.
\newblock {\em CoRR abs/1802.03888\/} (2018).

\bibitem{LL17}
{\sc Lundberg, S.~M., and Lee, S.-I.}
\newblock A unified approach to interpreting model predictions.
\newblock In {\em Advances in Neural Information Processing Systems 30},
  I.~Guyon, U.~V. Luxburg, S.~Bengio, H.~Wallach, R.~Fergus, S.~Vishwanathan,
  and R.~Garnett, Eds. Curran Associates, Inc., 2017, pp.~4768--4777.

\bibitem{Maleki2013}
{\sc Maleki, S., Tran-Thanh, L., Hines, G., Rahwan, T., and Rogers, A.}
\newblock Bounding the estimation error of sampling-based shapley value
  approximation with/without stratifying.
\newblock {\em ArXiv abs/1306.4265\/} (2013).

\bibitem{MudrakartaTSD18}
{\sc Mudrakarta, P.~K., Taly, A., Sundararajan, M., and Dhamdhere, K.}
\newblock Did the model understand the question?
\newblock In {\em Proceedings of the 56th Annual Meeting of the Association for
  Computational Linguistics, {ACL} 2018, Melbourne, Australia, July 15-20,
  2018, Volume 1: Long Papers\/} (2018), I.~Gurevych and Y.~Miyao, Eds.,
  Association for Computational Linguistics, pp.~1896--1906.

\bibitem{Owen2014}
{\sc Owen, A.~B.}
\newblock {Sobol'} indices and {Shapley} value.
\newblock {\em SIAM Journal of Uncertainty Quantification 2}, 1 (???? 2014),
  245--251.

\bibitem{Owen72}
{\sc Owen, G.}
\newblock Multilinear extensions of games.
\newblock {\em Management Science 18}, 5-part-2 (Jan. 1972), 64--79.

\bibitem{DBLP:conf/icml/2017}
{\sc Precup, D., and Teh, Y.~W.}, Eds.
\newblock {\em Proceedings of the 34th International Conference on Machine
  Learning, {ICML} 2017, Sydney, NSW, Australia, 6-11 August 2017\/} (2017),
  vol.~70 of {\em Proceedings of Machine Learning Research}, {PMLR}.

\bibitem{RajpurkarZLL16}
{\sc Rajpurkar, P., Zhang, J., Lopyrev, K., and Liang, P.}
\newblock Squad: 100, 000+ questions for machine comprehension of text.
\newblock In {\em Proceedings of the 2016 Conference on Empirical Methods in
  Natural Language Processing, {EMNLP} 2016, Austin, Texas, USA, November 1-4,
  2016\/} (2016), pp.~2383--2392.

\bibitem{Shapley53}
{\sc Shapley, L.~S.}
\newblock A value of n-person games.
\newblock {\em Contributions to the Theory of Games\/} (1953), 307--317.

\bibitem{SGSK16}
{\sc Shrikumar, A., Greenside, P., and Kundaje, A.}
\newblock Learning important features through propagating activation
  differences.
\newblock In Precup and Teh \cite{DBLP:conf/icml/2017}, pp.~3145--3153.

\bibitem{SVZ13}
{\sc Simonyan, K., Vedaldi, A., and Zisserman, A.}
\newblock Deep inside convolutional networks: Visualising image classification
  models and saliency maps.
\newblock {\em CoRR\/} (2013).

\bibitem{murdoch}
{\sc Singh, C., Murdoch, W.~J., and Yu, B.}
\newblock Hierarchical interpretations for neural network predictions.
\newblock In {\em International Conference on Learning Representations\/}
  (2019).

\bibitem{Sliwinski2019}
{\sc Sliwinski, J., Strobel, M., and Zick, Y.}
\newblock Axiomatic characterization of data-driven influence measures for
  classification.
\newblock In {\em The Thirty-Third {AAAI} Conference on Artificial
  Intelligence\/} (2019), pp.~718--725.

\bibitem{SDBR14}
{\sc Springenberg, J.~T., Dosovitskiy, A., Brox, T., and Riedmiller, M.~A.}
\newblock Striving for simplicity: The all convolutional net.
\newblock {\em CoRR\/} (2014).

\bibitem{Strumbelj2010}
{\sc Strumbelj, E., and Kononenko, I.}
\newblock An efficient explanation of individual classifications using game
  theory.
\newblock {\em J. Mach. Learn. Res. 11\/} (Mar. 2010), 1–18.

\bibitem{STY17}
{\sc Sundararajan, M., Taly, A., and Yan, Q.}
\newblock Axiomatic attribution for deep networks.
\newblock In Precup and Teh \cite{DBLP:conf/icml/2017}, pp.~3319--3328.

\bibitem{Lasso}
{\sc Tibshirani, R.}
\newblock Regression shrinkage and selection via the lasso.
\newblock {\em Journal of the Royal Statistical Society: Series B
  (Methodological) 58}, 1 (1996), 267--288.

\bibitem{Tsang}
{\sc Tsang, M., Cheng, D., and Liu, Y.}
\newblock Detecting statistical interactions from neural network weights.
\newblock In {\em International Conference on Learning Representations\/}
  (2018).

\bibitem{Tsang2}
{\sc Tsang, M., Liu, H., Purushotham, S., Murali, P., and Liu, Y.}
\newblock Neural interaction transparency (nit): Disentangling learned
  interactions for improved interpretability.
\newblock In {\em Advances in Neural Information Processing Systems 31},
  S.~Bengio, H.~Wallach, H.~Larochelle, K.~Grauman, N.~Cesa-Bianchi, and
  R.~Garnett, Eds. Curran Associates, Inc., 2018, pp.~5804--5813.

\bibitem{YuDLLZC2018}
{\sc Yu, A.~W., Dohan, D., Le, Q., Luong, T., Zhao, R., and Chen, K.}
\newblock Fast and accurate reading comprehension by combining self-attention
  and convolution.
\newblock In {\em International Conference on Learning Representations\/}
  (2018).

\end{thebibliography}
